\newif\iftechreport%
\techreporttrue%
\iftechreport%
\documentclass[a4paper]{article}
\else%
\documentclass[orivec]{llncs}
\fi%

\usepackage{hyperref}

\usepackage{amsmath}
\usepackage{amssymb}
\usepackage{stmaryrd}
\usepackage{mathpartir}
\usepackage{enumerate}
\usepackage{microtype}

\usepackage[dvipsnames]{xcolor}
\usepackage{todonotes}
\usepackage{soul}
\newcommand{\propdel}[1]{}

\usepackage{tikz}
\usetikzlibrary{positioning,arrows}
\usetikzlibrary{decorations,decorations.markings}

\newcommand{\buchi}{B\"{u}chi}

\newcommand{\naturals}{\mathbb{N}}
\newcommand{\angl}[1]{\left\langle#1\right\rangle}
\newcommand{\myrightarrow}[1]{\mathrel{\raisebox{-3pt}{$\xrightarrow{#1}$}}}
\newcommand{\csbot}{\mathbf{0}}
\newcommand{\cstop}{\mathbf{1}}
\newcommand{\abscsemiring}{\mathbb{E}}

\newcommand{\wcsemiring}{\mathbb{W}}
\newcommand{\composable}{\mathop{\ocircle}\nolimits}
\newcommand{\compose}{\mathop{\boxbox}\nolimits}

\usepackage{centernot}
\usepackage{paralist}
\usepackage{algorithm2e}

\iftechreport%
\usepackage[T1]{fontenc}
\usepackage{lmodern}

\usepackage{authblk}
\usepackage{amsthm}
\usepackage{hyperref}

\usepackage[margin=30mm]{geometry}

\newtheorem{theorem}{Theorem}
\newtheorem{lemma}{Lemma}
\newtheorem{proposition}{Proposition}
\newtheorem{corollary}{Corollary}

\theoremstyle{definition}

\newtheorem{definition}{Definition}
\fi%

\pgfmathsetmacro{\gap}{3}
\pgfmathsetmacro{\angle}{8}
\pgfmathsetmacro{\nodesize}{10mm}
\pgfmathsetmacro{\scale}{0.8}

\newcommand{\obtain}{\mathsf{obtain}}
\newcommand{\release}{\mathsf{release}}

\newcommand{\transfer}{\mathsf{transfer}}

\newcommand{\north}{\mathsf{north}}
\newcommand{\south}{\mathsf{south}}
\newcommand{\move}{\mathsf{move}}
\newcommand{\snapshot}{\mathsf{snapshot}}
\newcommand{\halt}{\mathsf{halt}}
\newcommand{\s}{\mathsf{s}}
\newcommand{\discharge}{\mathsf{discharge}}
\newcommand{\charge}{\mathsf{charge}}
\newcommand{\pass}{\mathsf{pass}}

\newcommand{\e}{\mathsf{e}}

\newcommand{\w}{\mathsf{w}}
\newcommand{\cc}{\mathsf{c}}

\newcommand{\dash}{\;|\:}

\newcommand{\until}{\mathop{U}}
\newcommand{\nxt}{\mathop{X}\nolimits}
\newcommand{\eventually}{\mathop{\Diamond}}
\newcommand{\always}{\mathop{\Box}}
\newcommand{\true}{\raisebox{-0.08em}{\ensuremath{\top}}}
\newcommand{\captures}{\mathop{\Yright}}
\renewcommand{\implies}{\mathop{\rightarrow}}

\newcommand{\bigbowtie}{\text{\Large $\bowtie$}}
\title{A Component-oriented Framework for~Autonomous Agents}

\date{}
\iftechreport%
\author[1]{Tobias Kapp\'e\thanks{\texttt{tkappe@cs.ucl.ac.uk}}}
\author[2,3]{Farhad Arbab}
\author[4]{Carolyn Talcott}
\affil[1]{University College London, London, United Kingdom}
\affil[2]{Centrum Wiskunde \& Informatica, Amsterdam, The Netherlands}
\affil[3]{LIACS, Leiden University, Leiden, The Netherlands}
\affil[4]{SRI International, Menlo Park, USA}
\else%
\usepackage[misc]{ifsym}
\author{%
    Tobias Kapp\'e\inst{1}\textsuperscript{(\Letter)}
    \and
    Farhad Arbab\inst{2,3}
    \and
    Carolyn Talcott\inst{4}
}

\institute{%
    University College London, London, United Kingdom \\
    \texttt{tkappe@cs.ucl.ac.uk}
    \and
    Centrum Wiskunde \& Informatica, Amsterdam, The Netherlands
    \and
    LIACS, Leiden University, Leiden, The Netherlands
    \and
    SRI International, Menlo Park, USA
}
\fi

\iftechreport\else%
\usepackage{environ}
\NewEnviron{killcontents}{}
\fi

\begin{document}

\maketitle

\begin{abstract}
The design of a complex system warrants a compositional methodology, i.e., composing simple components to obtain a larger system that exhibits their collective behavior in a meaningful way. We propose an automaton-based paradigm for compositional design of such systems where an \emph{action} is accompanied by one or more \emph{preferences}. At run-time, these preferences provide a natural fallback mechanism for the component, while at design-time they can be used to reason about the behavior of the component in an uncertain physical world. Using structures that tell us how to compose preferences and actions, we can compose formal representations of individual components or agents to obtain a representation of the composed system. We extend Linear Temporal Logic with two unary connectives that reflect the compositional structure of the actions, and show how it can be used to diagnose undesired behavior by tracing the falsification of a specification back to one or more culpable components.
\end{abstract}

\section{Introduction}
Consider the design of a software package that steers a crop surveillance drone. Such a system (in its simplest form, a single drone agent) should survey a field and relay the locations of possible signs of disease to its owner. There are a number of concerns at play here, including but not limited to maintaining an acceptable altitude, keeping an eye on battery levels and avoiding birds of prey. In such a situation, it is best practice to isolate these separate concerns into different modules --- thus allowing for code reuse, and requiring the use of well-defined protocols in case coordination between modules is necessary. One would also like to verify that the designed system satisfies desired properties, such as ``even on a conservative energy budget, the drone can always reach the charging station''.

In the event that the designed system violates its verification requirements or exhibits behavior that does not conform to the specification, it is often useful to have an example of such behavior. For instance, if the surveillance drone fails to maintain its target altitude, an example of behavior where this happens could tell us that the drone attempted to reach the far side of the field and ran out of energy. Additionally, failure to verify an LTL-like formula typically comes with a counterexample --- indeed, a counterexample arises from the automata-theoretic verification approach quite naturally~\cite{vardi-1995}. Taking this idea of \emph{diagnostics} one step further in the context of a compositional design, it would also be useful to be able to identify the components responsible for allowing a behavior that deviates from the specification, whether this behavior comes from a run-time observation or a design-time counterexample to a desired property. The designer then knows which components should be adjusted (in our example, this may turn out to be the route planning component), or, at the very least, rule out components that are not directly responsible (such as the wildlife evasion component).

In this paper, we propose an automata-based paradigm based on Soft Constraint Automata~\cite{arbab-santini-2013,kappe-arbab-talcott-2016}, called Soft Component Automata (SCAs\footnote{Here, we use the abbreviation \emph{SCA} exclusively to refer to Soft \emph{Component} Automata.}). An SCA is a state-transition system where transitions are labeled with actions and preferences. Higher-preference transitions typically contribute more towards the goal of the component; if a component is in a state where it wants the system to move north, a transition with action $\north$ has a higher preference than a transition with action $\south$. At run-time, preferences provide a natural fallback mechanism for an agent: in ideal circumstances, the agent would perform only actions with the highest preferences, but if the most-preferred actions fail, the agent may be permitted to choose a transition of lower preference. At design-time, preferences can be used to reason about the behavior of the SCA in suboptimal conditions, by allowing all actions whose preference is bounded from below by a threshold. In particular, this is useful if the designer wants to determine the circumstances (i.e., threshold on preferences) where a property is no longer verified by the system.

Because the actions and preferences of an SCA reside in well-defined mathematical structures, we can define a composition operator on SCAs that takes into account the composition of actions as well as preferences. The result of composition of two SCAs is another SCA where actions and preferences reflect those of the operands. As we shall see, SCAs are amenable to verification against formulas in Linear Temporal Logic (LTL). More specifically, one can check whether the behavior of an SCA is contained in the behavior allowed by a formula of LTL\@. 

Soft Component Automata are a generalization of Constraint Automata~\cite{baier-sirjani-arbab-rutten-2006}. The latter can be used to coordinate interaction between components in a verifiable fashion~\cite{baier-blechmann-klein-kluppelholz-leister-2010}. Just like Constraint Automata, the framework we present blurs the line between \emph{computation} and \emph{coordination} --- both are captured by the same type of automata. Consequently, this approach allows us to reason about these concepts in a uniform fashion: coordination is not separate in the model, it is effected by components which are inherently part of the model.

We present two contributions in this paper. First, we propose an compositional automata-based design paradigm for autonomous agents that contains enough information about actions to make agents behave in a robust manner --- by which we mean that, in less-than-ideal circumstances, the agent has alternative actions available when its most desired action turns out to be impossible, which help it achieve some subset of goals or its original goals to a lesser degree. We also put forth a dialect of LTL that accounts for the compositional structure of actions and can be used to verify guarantees about the behavior of components, as well as their behavior in composition. Our second contribution is a method to trace errant behavior back to one or more components, exploiting the algebraic structure of preferences. This method can be used with both run-time and design-time failures: in the former case, the behavior arises from the action history of the automaton, in the latter case it is a counterexample obtained from verification.

\medskip
In Section~\ref{section:related-work}, we mention some work related to this paper; in Section~\ref{section:preliminaries} we discuss the necessary notation and mathematical structures. In Section~\ref{section:component-model}, we introduce Soft Component Automata, along with a toy model. We discuss the syntax and semantics of the LTL-like logic used to verify properties of SCAs in Section~\ref{section:linear-temporal-logic}. In Section~\ref{section:diagnostics}, we propose a method to extract which components bear direct responsibility for a failure. Our conclusions comprise Section~\ref{section:conclusion}, and some directions for further work appear in Section~\ref{section:further-work}. 

\paragraph{Acknowledgements}
The authors would like to thank Vivek Nigam and the anonymous FACS-referees for their valuable feedback. This work was partially supported by  ONR grant N00014--15--1--2202.

\section{Related Work}%
\label{section:related-work}

The algebraic structure for preferences called the \emph{Constraint Semiring} was proposed by Bistarelli et al.~\cite{bistarelli-montanari-rossi-1995,bistarelli-2004}. Further exploration of the compositionality of such structures appears in~\cite{gadducci-holzl-monreale-wirsing-2013,holzl-meier-wirsing-2009,kappe-arbab-talcott-2016}. The structure we propose for modeling actions and their compositions is an algebraic reconsideration of \emph{static constructs}~\cite{huttel-larsen-1989}.

The automata formalism used in this paper generalizes \emph{Soft Constraint Automata}~\cite{baier-sirjani-arbab-rutten-2006,arbab-santini-2013}. The latter were originally proposed to give descriptions of Web Services~\cite{arbab-santini-2013}; in~\cite{kappe-arbab-talcott-2016}, they were used to model fault-tolerant, compositional autonomous agents. Using preference values to specify the behavior of autonomous agents is also explored from the perspective of rewriting logic in the \emph{Soft Agent Framework}~\cite{talcott-arbab-yadav-2015,talcott-nigam-arbab-kappe-2016}. Recent experiments with the Soft Agent Framework show that behavior based on soft constraints can indeed contribute robustness~\cite{mason-nigam-talcott-brito-2017}.

Sampath et al.~\cite{sampath-sengupta-lafortune-sinnamohideen-teneketzis-1996} discuss methods to detect unobservable errors based on a model of the system and a trace of observable events; others extended this approach~\cite{debouk-lafortune-teneketzis-2000,neidig-lunze-2005} to a multi-component setting. Casanova et al.~\cite{casanova-garlan-schmerl-abreu-2014} wrote about fault localisation in a system where some components are inobservable, based on which computations (tasks involving multiple components) fail. In these paradigms, one tries to find out where a \emph{runtime fault} occurs; in contrast, we try to find out which component is responsible for \emph{undesired behavior}, i.e., behavior that is allowed by the system but not desired by the specification.

A general framework for fault ascription in concurrent systems based on \emph{counterfactuals} is presented in~\cite{goessler-astefanoaei-2014,goessler-stefani-2015}. Formal definitions are given for failures in a given set of components to be necessary and/or sufficient cause of a system violating a given property. Components are specified by sets of sets of events (analogous to actions) representing possible correct behaviors. A parallel (asynchronous) composition operation is defined on components, but there is no notion of composition of events or explicit interaction between components. A system is given by a global behavior (a set of event sets) together with a set of system component specifications. The global behavior, which must be provided separately, includes component events, but may also have other events, and may violate component specifications (hence the faulty components).  In our approach, global behavior is  obtained by component composition. Undesired behavior may be local to a component or emerge as the result of interactions.

In LTL, a counterexample to a negative result arises naturally if one employs automata-based verification techniques~\cite{muller-saoudi-schupp-1988,vardi-1995}. In this paper, we further exploit counterexamples to gain information about the component or components involved in violating the specification. The application of LTL to Constraint Automata is inspired by an earlier use of LTL for Constraint Automata~\cite{baier-blechmann-klein-kluppelholz-leister-2010}.

Some material in this paper appeared in the first author's master's thesis~\cite{kappe-2016-thesis}.

\section{Preliminaries}%
\label{section:preliminaries}

If $\Sigma$ is a set, then $2^\Sigma$ denotes the set of subsets of $\Sigma$, i.e., the \emph{powerset} of $\Sigma$. We write $\Sigma^*$ for the set of \emph{finite words} over $\Sigma$, and if $\sigma \in \Sigma^*$ we write $|\sigma|$ for the \emph{length} of $\sigma$. We write $\sigma(n)$ for the $n$-th letter of $\sigma$ (starting at $0$). Furthermore, let $\Sigma^\omega$ denote the set of functions from $\mathbb{N}$ to $\Sigma$, also known as \emph{streams} over $\Sigma$~\cite{rutten-2005}. We define for $\sigma \in \Sigma^\omega$ that $|\sigma| = \omega$ (the smallest infinite ordinal). Concatenation of a stream to a finite word is defined as expected. We use the superscript $\omega$ to denote infinite repetition, writing $\sigma = \angl{0,1}^\omega$ for the parity function; we write $\Sigma^\pi$ for the set of \emph{eventually periodic} streams in $\Sigma^\omega$, i.e., $\sigma \in \Sigma^\omega$ such that there exist $\sigma_h, \sigma_t \in \Sigma^*$ with $\sigma = \sigma_h \cdot \sigma_t^\omega$. We write $\sigma^{(k)}$ with $k \in \naturals$ for the \emph{$k$-th derivative} of $\sigma$, which is given by $\sigma^{(k)}(n) = \sigma(k+n)$.

If $S$ is a set and $\odot: S \times S \to S$ a function, we refer to $\odot$ as an \emph{operator on $S$} and write $p \odot q$ instead of $\odot(p,q)$. We always use parentheses to disambiguate expressions if necessary. To model composition of actions, we need a slight generalization. If $R \subseteq S \times S$ is a relation and $\odot: R \to S$ is a function, we refer to $\odot$ as a \emph{partial operator on $S$ up to $R$}; we also use infix notation by writing $p \odot q$ instead of $\odot(p,q)$ whenever $pRq$. If $\odot: R \to S$ is a partial operator on $S$ up to $R$, we refer to $\odot$ as \emph{idempotent} if $p \odot p = p$ for all $p \in S$ such that $pRp$, and \emph{commutative} if $p \odot q = q \odot p$ whenever $p,q \in S$, $pRq$ and $qRp$. Lastly, $\odot$ is \emph{associative} if for all $p, q, r \in S$, $p R q$ and $(p \odot q) R r$ if and only if $q R r$ and $p R (q \odot r)$, either of which implies that $(p \odot q) \odot r = p \odot (q \odot r)$. When $R = S \times S$, we recover the canonical definitions of idempotency, commutativity and associativity.

A \emph{constraint semiring}, or \emph{c-semiring}, provides a structure on preference values that allows us to \emph{compare} the preferences of two actions to see if one is preferred over the other as well as \emph{compose} preference values of component actions to find out the preference of their composed action. A c-semiring~\cite{bistarelli-montanari-rossi-1995,bistarelli-2004} is a tuple $\angl{\abscsemiring, \bigoplus, \otimes, \csbot, \cstop}$ such that
\begin{inparaenum}[(1)]
    \item $\abscsemiring$ is a set, called the \emph{carrier}, with $\csbot, \cstop \in \abscsemiring$,
    \item $\bigoplus: 2^\abscsemiring \to \abscsemiring$ is a function such that for $e \in \abscsemiring$ we have that $\bigoplus \emptyset = \csbot$ and $\bigoplus \abscsemiring = \cstop$, as well as $\bigoplus \{ e \} = e$, and for $\mathcal{E} \subseteq 2^\abscsemiring$, also $\bigoplus \left\{ \bigoplus(E) : E \in \mathcal{E} \right\} = \bigoplus \bigcup \mathcal{E}$ (the \emph{flattening property}), and
    \item $\otimes: \abscsemiring \times \abscsemiring \to \abscsemiring$ is a commutative and associative operator, such that for $e \in \abscsemiring$ and $E \subseteq \abscsemiring$, it holds that $e \otimes \csbot = \csbot$ and $e \otimes \cstop = e$ as well as $e \otimes \bigoplus E = \bigoplus \{ e \otimes e' : e' \in E \}$.
\end{inparaenum}
We denote a c-semiring by its carrier; if we refer to $\abscsemiring$ as a c-semiring, associated symbols are denoted $\bigoplus_\abscsemiring, \csbot_\abscsemiring$, et cetera. We drop the subscript when only one c-semiring is in context.

The operator $\bigoplus$ of a c-semiring $\abscsemiring$ induces an idempotent, commutative and associative binary operator $\oplus: \abscsemiring \times \abscsemiring \to \abscsemiring$ by defining $e \oplus e' = \bigoplus (\{ e, e' \})$
The relation $\leq_{\abscsemiring}\ \subseteq \abscsemiring \times \abscsemiring$ is such that $e \leq_{\abscsemiring} e'$ if and only if $e \oplus e' = e'$; $\leq_{\abscsemiring}$ is a partial order on $\abscsemiring$, with $\csbot$ and $\cstop$ the minimal and maximal elements~\cite{bistarelli-2004}. All c-semirings are complete lattices, with $\bigoplus$ filling the role of the least upper bound operator~\cite{bistarelli-2004}. Furthermore, $\otimes$ is \emph{intensive}, meaning that for any $e, e' \in \abscsemiring$, we have $e \otimes e' \leq e$~\cite{bistarelli-2004}. Lastly, when $\otimes$ is idempotent, $\otimes$ coincides with the greatest lower bound~\cite{bistarelli-2004}. 

Models of a c-semiring include $\wcsemiring = \angl{\mathbb{R}_{\geq 0} \cup \{ \infty \}, \inf, \hat{+}, \infty, 0}$ (the \emph{weighted semiring}), where $\inf$ is the infimum and $\hat{+}$ is arithmetic addition generalized to $\mathbb{R}_{\geq 0} \cup \{ \infty \}$. Here, $\leq_{\wcsemiring}$ coincides with the obvious definition of the order $\geq$ on $\mathbb{R}_{\geq 0} \cup \{ \infty \}$. Composition operators for c-semirings exist, such as product composition~\cite{bistarelli-montanari-rossi-1997} and (partial) lexicographic composition~\cite{gadducci-holzl-monreale-wirsing-2013}. We refer to~\cite{kappe-arbab-talcott-2016} for a self-contained discussion of these composition techniques.

\section{Component Model}%
\label{section:component-model}

We now discuss our component model for the construction of autonomous agents. 

\subsection{Component Action Systems}
Observable behavior of agents is the result of the actions put forth by their individual components; we thus need a way to talk about how actions compose. For example, in our crop surveillance drone, the following may occur:
\begin{itemize}
    \item The component responsible for taking pictures wants to take a snapshot, while the routing component wants to move north. Assuming the camera is capable of taking pictures while moving, these actions may compose into the action ``take a snapshot while moving north''. In this case, actions compose \emph{concurrently}, and we say that the latter action \emph{captures} the former two. 
    \item The drone has a single antenna that can be used for GPS and communications, but not both at the same time. The component responsible for relaying pictures has finished its transmission and wants to release its lock on the antenna, while the navigation component wants to get a fix on the location and requests use of the antenna. In this case, the actions ``release privilege'' and ``obtain privilege'' compose \emph{logically}, into a ``transfer privilege'' action.
    \item The routing component wants to move north, while the wildlife avoidance component notices a hawk approaching from that same direction, and thus wants to move south. In this case, the intentions of the two components are contradictory; these component actions are \emph{incomposable}, and some resolution mechanism (e.g., priority) will have to decide which action takes precedence.
\end{itemize}
All of these possibilities are captured in the definition below.
\begin{definition}
A \emph{Component Action System (CAS)} is a tuple $\angl{\Sigma, \composable, \compose}$, such that $\Sigma$ is a finite set of \emph{actions}, $\composable \subseteq \Sigma \times \Sigma$ is a reflexive and symmetric relation and $\compose: \composable \to \Sigma$ is an idempotent, commutative and associative operator on $\Sigma$ up to $\composable$ (i.e., $\compose$ is an operator defined only on elements of $\Sigma$ related by $\composable$). We call $\composable$ the \emph{composability relation}, and $\compose$ the \emph{composition operator}.
\end{definition}
Every CAS $\angl{\Sigma, \composable, \compose}$ induces a relation $\sqsubseteq$ on $\Sigma$, where for $a, b \in \Sigma$, $a \sqsubseteq b$ if and only if there exists a $c \in \Sigma$ such that $a$ and $c$ are composable ($a \composable c$) and they compose into $b$ ($a \compose c = b$). One can easily verify that $\sqsubseteq$ is a preorder; accordingly, we call $\sqsubseteq$ the \emph{capture preorder} of the CAS\@.

As with c-semirings, we may refer to a set $\Sigma$ as a CAS\@. When we do, its composability relation, composition operator and preorder are denoted by $\composable_\Sigma$, $\compose_\Sigma$ and $\sqsubseteq_\Sigma$. We drop the subscript when there is only one CAS in context.

We model incomposability of actions by omitting them from the composability relation; i.e., if $\south$ is an action that compels the agent to move south, while $\north$ drives the agent north, we set $\south \centernot\composable \north$. Note that $\composable$ is not necessarily transitive. This makes sense in the scenarios above, where $\snapshot$ is composable with $\south$ as well as $\north$, but $\north$ is incomposable with $\south$. Moreover, incomposability carries over to compositions: if $\south \composable \snapshot$ and $\south \centernot\composable \north$, also $(\south \compose \snapshot) \centernot\composable \north$. This is formalized in the following lemma.
\begin{lemma}
Let $\angl{\Sigma, \composable, \compose}$ be a CAS and let $a, b, c \in \Sigma$. If $a \composable b$ but $a \centernot\composable c$, then $(a \compose b) \centernot\composable c$. Moreover, if $a \centernot\composable c$ and $a \sqsubseteq b$, then $b \centernot\composable c$.
\end{lemma}
\begin{proof}
For the first claim, suppose that $(a \compose b) \composable c$. Then, since $\compose$ is associative up to $\composable$, it follows that $b \composable c$ and $a \composable (b \compose c)$, which contradicts the premise that $b \centernot\composable c$. We thus conclude that $(a \compose b) \centernot\composable c$. 

For the second claim, suppose that $a \centernot\composable c$ and $a \sqsubseteq b$. Then there exists a $d \in \Sigma$ such that $a \composable d$ and $a \compose d = b$. By the above, $b = (a \compose d) \centernot\composable c$.
\end{proof}

The composition operator facilitates concurrent as well as logical composition. Given actions $\obtain$, $\release$ and $\transfer$, with their interpretation as in the second scenario, we can encode that $\obtain$ and $\release$ are composable by stipulating that $\obtain \composable \release$, and say that their (logical) composition involves an exchange of privileges by choosing $\obtain \compose \release = \transfer$. Furthermore, the capture preorder describes our intuition of capturing: if $\snapshot$ and $\move$ are the actions of the first scenario, with $\snapshot \composable \north$, then $\snapshot, \north \sqsubseteq \snapshot \compose \north$.

Port Automata~\cite{koehler-clarke-2009} contain a model of a CAS\@. Here, actions are sets of symbols called \emph{ports}, i.e., elements of $2^P$ for some finite set $P$. Actions $\alpha, \beta \in 2^P$ are compatible when they agree on a fixed set $\gamma \subseteq P$, i.e., if $\alpha \cap \gamma = \beta \cap \gamma$, and their composition is $\alpha \cup \beta$. Similarly, we also find an instance of a CAS in \emph{(Soft) Constraint Automata}~\cite{baier-sirjani-arbab-rutten-2006,arbab-santini-2013}; see~\cite{kappe-2016-thesis} for a full discussion of this correspondence.

\subsection{Soft Component Automata}
Having introduced the structure we impose on actions, we are now ready to discuss the automaton formalism that specifies the sequences of actions that are allowed, along with the preferences attached to such actions.

\begin{definition}
A \emph{Soft Component Automaton (SCA)} is a tuple $\angl{Q, \mskip-1mu\Sigma, \abscsemiring, \rightarrow, q^0, t}$ where $Q$ is a finite set of \emph{states}, with $q^0 \in Q$ the \emph{initial state}, $\Sigma$ is a CAS and $\abscsemiring$ is a c-semiring with $t \in \abscsemiring$, and $\rightarrow\ \subseteq Q \times \Sigma \times \abscsemiring \times Q$ is a finite relation called the \emph{transition relation}. We write $q \myrightarrow{a,\, e} q'$ when $\angl{q, a, e, q'} \in\ \rightarrow$.
\end{definition}
An SCA models the actions available in each state of the component, how much these actions contribute towards the goal and the way actions transform the state. The threshold value restricts the available actions to those with a preference bounded from below by the threshold, either at run-time, or at design-time when one wants to reason about behaviors satisfying some minimum preference.

We stress here that the threshold value is purposefully defined as part of an SCA, rather than as a parameter to the semantics in Section~\ref{section:behavioral-semantics}. This allows us to speak of the preferences of an individual component, rather than a threshold imposed on the whole system; instead, the threshold of the system arises from the thresholds of the components, which is especially useful in Section~\ref{section:diagnostics}.


We depict SCAs in a fashion similar to the graphical representation of finite state automata: as a labeled graph, where vertices represent states and the edges transitions, labeled with elements of the CAS and c-semiring. The initial state is indicated by an arrow without origin. The CAS, c-semiring and threshold value will always be made clear where they are germane to the discussion.

An example of an SCA is $A_\e$, drawn in Figure~\ref{figure:sca-energy}; its CAS contains the incomposable actions $\charge$, $\discharge_1$ and $\discharge_2$, and its c-semiring is the weighted semiring $\wcsemiring$. This particular SCA can model the component of the crop surveillance drone responsible for keeping track of the amount of energy remaining in the system; in state $q_n$ (for $n \in \{0,1,\dots,4\}$), the drone has $n$ units of energy left, meaning that in states $q_1$ to $q_4$, the component can spend one unit of energy through $\discharge_1$, and in states $q_2$ to $q_4$, the drone can consume two units of energy through $\discharge_2$. In states $q_0$ to $q_3$, the drone can try to recharge through $\charge$.%
\footnote{This is a rather simplistic description of energy management. We remark that a more detailed description is possible by extending SCAs with \emph{memory cells}~\cite{jongmans-kappe-arbab-2017} and using a memory cell to store the energy level. In such a setup, a state would represent a \emph{range} of energy values that determines the components disposition regarding resources.}
Recall that, in $\wcsemiring$, higher values reflect a lower preference (a higher \emph{weight}); thus, $\charge$ is preferred over $\discharge_1$.
\begin{figure}[ht!]
    \centering
    \begin{tikzpicture}[scale=\scale, transform shape]
        \begin{scope}[every node/.style={draw,circle},minimum size=\nodesize]
            \node (e0) at (0,0) {$q_0$};
            \node (e1) at (3,0) {$q_1$};
            \node (e2) at (6,0) {$q_2$};
            \node (e3) at (9,0) {$q_3$};
            \node (e4) at (12,0) {$q_4$};
        \end{scope}

        \node[right=5mm of e4] (s) {};
        \path[->] (s) edge (e4);

        \begin{scope}[%
            every path/.style={->,looseness=0.5},%
            every node/.style={align=left}%
        ]
            \path (e0.\angle) edge[bend left] node[above] {$\charge, 0$} (e1.180-\angle);
            \path (e1.180+\angle) edge[bend left] node[below] {$\discharge_1, 2$} (e0.-\angle);
            \path (e1.\angle) edge[bend left] node[above] {$\charge, 0$} (e2.180-\angle);
            \path (e2.180+\angle) edge[bend left] node[below] {$\discharge_1, 2$} (e1.-\angle);
            \path (e2.\angle) edge[bend left] node[above] {$\charge, 0$} (e3.180-\angle);
            \path (e3.180+\angle) edge[bend left] node[below] {$\discharge_1, 2$} (e2.-\angle);
            \path (e3.\angle) edge[bend left] node[above] {$\charge, 0$} (e4.180-\angle);
            \path (e4.180+\angle) edge[bend left] node[below] {$\discharge_1, 2$} (e3.-\angle);
            \path (e2.south) edge[bend left,looseness=0.7] node[below] {$\discharge_2, 5$} (e0.south);
            \path (e3.south) edge[bend left,looseness=0.7] node[below] {$\discharge_2, 5$} (e1.south);
            \path (e4.south) edge[bend left,looseness=0.7] node[below] {$\discharge_2, 5$} (e2.south);
        \end{scope}
    \end{tikzpicture}
    \caption{A component modeling energy management, $A_\e$.}\label{figure:sca-energy}
\end{figure}
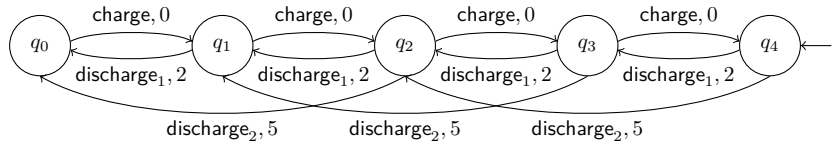

Here, $A_\e$ is meant to describe the possible behavior of the energy management component only. Availability of the actions within the \emph{total model} of the drone (i.e., the composition of all components) is subject to how actions compose with those of other components; for example, the availability of $\charge$ may depend on the state of the component modelling position. Similarly, preferences attached to actions concern energy management only. In states $q_0$ to $q_3$, the component prefers to top up its energy level through $\charge$, but the preferences of this component under composition with some other component may cause the composed preferences of actions composed with $\charge$ to be different. For instance, the total model may prefer executing an action that captures $\discharge_2$ over one that captures $\charge$ when the former entails movement and the latter does not, especially when survival necessitates movement. 

Nevertheless, the preferences of $A_\e$ affect the total behavior. For instance, the weight of spending one unit of energy (through $\discharge_1$) is lower than the weight of spending two units (through $\discharge_2$). This means that the energy component prefers to spend a small amount of energy before re-evaluating over spending more units of energy in one step. This reflects a level of care: by preferring small steps, the component hopes to avoid situations where too little energy is left to avoid disaster.

\subsection{Composition}

Composition of two SCAs arises naturally, as follows.
\begin{definition}
Let $A_i = \angl{Q_i, \Sigma, \abscsemiring, \rightarrow_i, q^0_i, t_i}$ be an SCA for $i \in \{0, 1\}$. The \emph{(parallel) composition} of $A_0$ and $A_1$ is the SCA $\angl{Q, \Sigma, \abscsemiring, \rightarrow, q^0, t_0 \otimes t_1}$, denoted $A_0 \bowtie A_1$, where $Q = Q_0 \times Q_1$, $q^0 = \angl{q^0_0, q^0_1}$, $\otimes$ is the composition operator of $\abscsemiring$, and $\rightarrow$ is the smallest relation satisfying
\[
\inferrule{%
    q_0 \myrightarrow{a_0,\ e_0}_0 q_0' \\
    q_1 \myrightarrow{a_1,\ e_1}_1 q_1' \\
    a_0 \composable a_1
}{%
    \angl{q_0, q_1} \myrightarrow{a_0 \compose a_1,\ e_0 \otimes e_1} \angl{q_0', q_1'}
}
\]
\end{definition}
In a sense, composition is a generalized product of automata, where composition of actions is mediated by the CAS\@: transitions with composable actions manifest in the composed automaton, as transitions with composed action and preference.

Composition is defined for SCAs that share CAS and c-semiring. Absent a common CAS, we do not know which actions compose, and what their compositions are. However, composition of SCAs with different c-semirings does make sense when the components model different concerns (e.g., for our crop surveillance drone, ``minimize energy consumed'' and ``maximize covering of snapshots''), both contributing towards the overall goal. Earlier work on Soft Constraint Automata~\cite{kappe-arbab-talcott-2016} explored this possibility. The additional composition operators proposed there can easily be applied to Soft Component Automata.


A state $q$ of a component may become unreachable after composition, in the sense that no state composed of $q$ is reachable from the composed initial state. For example, in the total model of our drone, it may occur that any state representing the drone at the far side of the field is unreachable, because the energy management component prevents some transition for lack of energy.

To discuss an example of SCA composition, we introduce the SCA $A_\s$ in Figure~\ref{figure:sca-snapshot}, which models the concern of the crop surveillance drone that it should take a snapshot of every location before moving to the next. The CAS of $A_\s$ includes the pairwise incomposable actions $\pass$, $\move$ and $\snapshot$, and its c-semiring is the weighted c-semiring $\wcsemiring$. We leave the threshold value $t_\s$ undefined for now. The purpose of $A_\s$ is reflected in its states: $q_Y$ (respectively $q_N$) represents that a snapshot of the current location was (respectively was not) taken since moving there. If the drone moves to a new location, the component moves to $q_N$, while $q_Y$ is reached by taking a snapshot. If the drone has not yet taken a snapshot, it prefers to do so over moving to the next spot (missing the opportunity).\footnote{A more detailed description of such a component could count the number of times the drone has moved without taking a snapshot first, and assign the preference of doing so again accordingly.}

\begin{figure}
    \centering
    \begin{tikzpicture}[scale=\scale, transform shape]
        \begin{scope}[every node/.style={draw,circle},minimum size=\nodesize]
            \node (sy) at (0,0) {$q_Y$};
            \node (sn) at (3,0) {$q_N$};
        \end{scope}

        \node[above right=5mm of sn] (s) {};
        \path[->] (s) edge (sn);

        \begin{scope}[%
            every path/.style={->,looseness=0.5},%
            every node/.style={align=left}%
        ]
            \path (sy.\angle) edge[bend left] node[above] {$\move, 0$} (sn.180-\angle);
            \path (sn.180+\angle) edge[bend left] node[below] {$\snapshot, 0$} (sy.-\angle);
            \path (sn) edge[loop right] node[right, align=left] {$\move, 2$ \\ $\pass, 1$} (sn);
            \path (sy) edge[loop left] node[left, align=right] {$\pass, 1$} (sy);
        \end{scope}
    \end{tikzpicture}
    \caption{A component modeling the desire to take a snapshot at every location, $A_\s$.}\label{figure:sca-snapshot}
\end{figure}
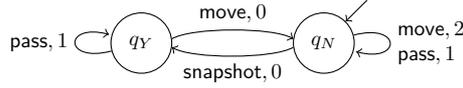

We grow the CAS of $A_\e$ and $A_\s$ to include the actions $\move$, $\move_2$, $\snapshot$ and $\snapshot_1$ (here, the action $\alpha_i$ is interpreted as ``execute action $\alpha$ and account for $i$ units of energy spent''), and $\composable$ is the smallest reflexive, commutative and transitive relation such that the following hold: $\move \composable \discharge_2$ (moving costs two units of energy), $\snapshot \composable \discharge_1$ (taking a snapshot costs one unit of energy) and $\pass \composable \charge$ (the snapshot state is unaffected by charging). We also choose $\move \compose \discharge_2 = \move_2$, $\snapshot \compose \discharge_1 = \snapshot_1$ and $\pass \compose \charge = \charge$. The composition of $A_\e$ and $A_\e$ is depicted in Figure~\ref{figure:sca-energy-snapshot}. 

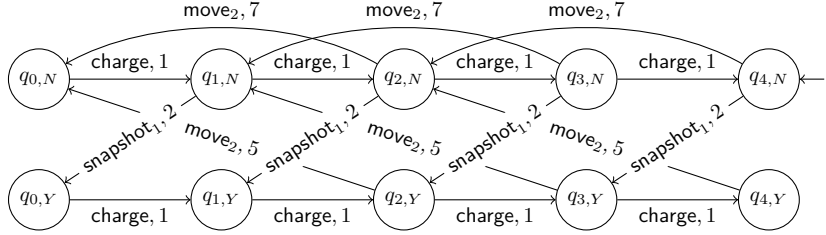
\begin{figure}
    \centering
    \begin{tikzpicture}[scale=\scale, transform shape]
        \begin{scope}[every node/.style={draw,circle},minimum size=\nodesize]
            \node (e0n) at (0,2) {$q_{0,N}$};
            \node (e1n) at (3,2) {$q_{1,N}$};
            \node (e2n) at (6,2) {$q_{2,N}$};
            \node (e3n) at (9,2) {$q_{3,N}$};
            \node (e4n) at (12,2) {$q_{4,N}$};
            \node (e0y) at (0,0) {$q_{0,Y}$};
            \node (e1y) at (3,0) {$q_{1,Y}$};
            \node (e2y) at (6,0) {$q_{2,Y}$};
            \node (e3y) at (9,0) {$q_{3,Y}$};
            \node (e4y) at (12,0) {$q_{4,Y}$};
        \end{scope}

        \node[right=4mm of e4n] (s) {};
        \path[->] (s) edge (e4n);

        \begin{scope}[%
            every path/.style={->},%
            every node/.style={align=left}%
        ]
            \path (e0n) edge node[above] {$\charge, 1$} (e1n);
            \path (e1n) edge node[above] {$\charge, 1$} (e2n);
            \path (e2n) edge node[above] {$\charge, 1$} (e3n);
            \path (e3n) edge node[above] {$\charge, 1$} (e4n);
            \path (e0y) edge node[below] {$\charge, 1$} (e1y);
            \path (e1y) edge node[below] {$\charge, 1$} (e2y);
            \path (e2y) edge node[below] {$\charge, 1$} (e3y);
            \path (e3y) edge node[below] {$\charge, 1$} (e4y);
            \path (e2y) edge[sloped] node[fill=white] {$\move_2, 5$} (e0n);
            \path (e3y) edge[sloped] node[fill=white] {$\move_2, 5$} (e1n);
            \path (e4y) edge[sloped] node[fill=white] {$\move_2, 5$} (e2n);
            \path (e1n) edge[sloped] node[fill=white] {$\snapshot_1, 2$} (e0y);
            \path (e2n) edge[sloped] node[fill=white] {$\snapshot_1, 2$} (e1y);
            \path (e3n) edge[sloped] node[fill=white] {$\snapshot_1, 2$} (e2y);
            \path (e4n) edge[sloped] node[fill=white] {$\snapshot_1, 2$} (e3y);
            \path (e2n) edge[bend right,looseness=0.8] node[above] {$\move_2, 7$} (e0n);
            \path (e3n) edge[bend right,looseness=0.8] node[above] {$\move_2, 7$} (e1n);
            \path (e4n) edge[bend right,looseness=0.8] node[above] {$\move_2, 7$} (e2n);
        \end{scope}
    \end{tikzpicture}
    \caption{The composition of the SCAs $A_\e$ and $A_\s$, dubbed $A_{\e, \s}$: a component modeling energy and snapshot management. We abbreviate pairs of states $\angl{q_i, q_j}$ by writing $q_{i,j}$.}\label{figure:sca-energy-snapshot}
\end{figure}

The structure of $A_{\e, \s}$ reflects that of $A_\e$ and $A_\s$; for instance, in state $q_{2,Y}$ two units of energy remain, and we have a snapshot of the current location. The same holds for the transitions of $A_{\e,\s}$; for example, $q_{2,N} \myrightarrow{\snapshot_1,\ 2} q_{1,Y}$ is the result of composing $q_2 \myrightarrow{\discharge_1,\ 2} q_1$ and $q_N \myrightarrow{\snapshot,\ 0} q_Y$. 

Also, note that in $A_{\e, \s}$ the preference of the $\move_2$-transitions at the top of the figure is lower than the preference of the diagonally-drawn $\move_2$-transitions. This difference arises because the component transition in $A_\s$ of the former is $q_N \myrightarrow{\move,\ 2} q_N$, while that of the latter is $q_Y \myrightarrow{\move,\ 0} q_N$. As such, the preferences of the component SCAs manifest in the preferences of the composed SCA\@.

The action $\snapshot_1$ is not available in states of the form $q_{i,Y}$, because the only action available in $q_Y$ is $\pass$, which does not compose into $\snapshot_1$.

\subsection{Behavioral semantics}%
\label{section:behavioral-semantics}

The final part of our component model is a description of the behavior of SCAs\@. Here, the threshold determines which actions have sufficient preference for inclusion in the behavior. Intuitively, the threshold is an indication of the amount of flexibility allowed. In the context of composition, lowering the threshold of a component is a form of compromise: the component potentially gains behavior available for composition. Setting a lower threshold makes a component more permissive, but may also make it harder (or impossible) to achieve its goal.

The question of where to set the threshold is one that the designer of the system should answer based on the properties and level of flexibility expected from the component; Section~\ref{section:linear-temporal-logic} addresses the formulation of these properties, while Section~\ref{section:diagnostics} talks about adjusting the threshold.

\begin{definition}%
\label{definition:language}
Let $A = \angl{Q, \Sigma, \abscsemiring, \rightarrow, q^0, t}$ be an SCA\@. We say that a stream $\sigma \in \Sigma^\omega$ is a \emph{behavior} of $A$ when there exist streams $\mu \in Q^\omega$ and $\nu \in \abscsemiring^\omega$ such that $\mu(0) = q^0$, and for all $n \in \naturals$, $t \leq \nu(n)$ and $\mu(n) \myrightarrow{\sigma(n),\ \nu(n)} \mu(n+1)$. The set of behaviors of $A$, denoted by $L(A)$, is called the \emph{language} of $A$.
\end{definition}

We note the similarity between the behavior of an SCA and that of \buchi-automata~\cite{buchi-1962}; we elaborate on this in 
\iftechreport%
Appendix~\ref{appendix:decision-procedure}.
\fi

To account for states that lack outgoing transitions, one could include implicit transitions labelled with $\halt$ (and some appropriate preference) to an otherwise unreachable ``halt state'', with a $\halt$ self-loop. Here, we set for all $\alpha \in \Sigma$ that $\halt \composable \alpha$ and $\halt \compose \alpha = \halt$. To simplify matters, we do not elaborate on this.

Consider $\sigma = \angl{\snapshot, \move, \move}^\omega$ and $\tau  = \angl{\snapshot, \move, \pass}^\omega$. We can see that when $t_\s = 2$, both are behaviors of $A_\s$; when $t_\s = 1$, $\tau$ is a behavior of $A_\s$, while $\sigma$ is not, since every second $\move$-action in $\sigma$ has preference $2$. More generally, if $A$ and $A'$ are SCAs over c-semiring $\abscsemiring$ that only differ in their threshold values $t, t' \in \abscsemiring$, and $t \leq t'$, then $L(A') \subseteq L(A)$. In the case of $A_\e$, the threshold can be interpreted as a bound on the amount of energy to be spent in a single action; if $t_\e < 5$, then behaviors with $\discharge_2$ do not occur in $L(A_\e)$. 

Interestingly, if $A_1$ and $A_2$ are SCAs, then $L(A_1 \bowtie A_2)$ is not uniquely determined by $L(A_1)$ and $L(A_2)$. For example, suppose that $t_\e = 4$ and $t_\s = 1$, and consider $L(A_{\e, \s})$, which contains $\angl{\snapshot} \cdot \angl{\move, \snapshot, \charge, \charge, \charge}^\omega$ even though the corresponding stream of component actions in $A_\e$, i.e., the stream $\angl{\discharge_1} \cdot \angl{\discharge_2, \discharge_1, \charge, \charge, \charge}^\omega$ is not contained in $L(A_\e)$. This is a consequence of a more general observation for c-semirings, namely that $t \leq e$ and $t' \leq e'$ is sufficient but not necessary to derive $t \otimes t' \leq e \otimes e'$.

\section{Linear Temporal Logic}%
\label{section:linear-temporal-logic}

We now turn our attention to verifying the behavior of an agent, by means of a simple dialect of Linear Temporal Logic (LTL). The aim of extending LTL is to reflect the compositional nature of the actions. This extension has two aspects, which correspond roughly to the relations $\sqsubseteq$ and $\composable$: reasoning about behaviors that \emph{capture} (i.e., are composed of) other behaviors, and about behaviors that are \emph{composable} with other behaviors. For instance, consider the following scenarios:
\begin{enumerate}[(i)]
    \item\label{scenario:captures} We want to verify that under certain circumstances, the drone performs a series of actions where it goes north before taking a snapshot. This is useful when, for this particular property, we do not care about other actions that may also be performed while or as part of going north, for instance, whether or not the drone engages in communications while moving.
    \item\label{scenario:composable} We want to verify that every behavior of the snapshot-component is composable with some behavior that eventually recharges. This is useful when we want to abstract away from the action that allows recharging, i.e., it is not important which particular action composes with $\charge$.
\end{enumerate}
Our logic aims to accommodate both scenarios, by providing two new connectives: $\captures \phi$ describes every behavior that captures a behavior validating $\phi$, while $\composable \phi$ holds for every behavior composable with a behavior validating $\phi$.

\subsection{Syntax and semantics}%
\label{section:syntax-and-semantics}

The syntax of the LTL dialect we propose for SCAs contains atoms, conjunctions, negation, and the ``until'' and ``next'' connectives, as well as the unary connectives $\composable$ and $\captures$. Formally, given a CAS $\Sigma$, the language $\mathcal{L}_\Sigma$ is generated by the grammar
\[
\phi, \psi ::= \top \dash a \in \Sigma \dash \phi \wedge \psi \dash \phi \until \psi \dash \nxt \phi \dash \neg \phi \dash \captures \phi \dash \composable \phi
\]
As a convention, unary connectives take precedence over binary connectives. For example, $\captures \phi \until \neg \psi$ should be read as $(\captures \phi) \until (\neg \psi)$. We use parentheses to disambiguate formulas where this convention does not give a unique bracketing.

The semantics of our logic is given as a relation $\models_\Sigma$ between $\Sigma^\omega$ and $\mathcal{L}_\Sigma$; to be precise, $\models_\Sigma$ is the smallest such relation that satisfies the following rules
\[
\inferrule{%
    \sigma \in \Sigma^\omega
}{%
    \sigma \models_\Sigma \true
}
\quad
\inferrule{%
    \sigma \in \Sigma^\omega
}{%
    \sigma \models_\Sigma \sigma(0)
}
\quad
\inferrule{%
    \sigma \models_\Sigma \phi \\
    \sigma \models_\Sigma \psi
}{%
    \sigma \models_\Sigma \phi \wedge \psi
}
\]
\[
\inferrule{%
    n \in \naturals \\
    \forall k < n.\ \sigma^{(k)} \models_\Sigma \phi \\
    \sigma^{(n)} \models_\Sigma \psi
}{%
    \sigma \models_\Sigma \phi \until \psi
}
\quad
\inferrule{%
    \sigma^{(1)} \models_\Sigma \phi
}{%
    \sigma \models_\Sigma \nxt \phi
}
\]
\[
\inferrule{%
    \sigma \not\models_\Sigma \phi
}{%
    \sigma \models_\Sigma \neg\phi
}
\quad
\inferrule{%
    \sigma \models_\Sigma \phi \\
    \sigma \sqsubseteq^\omega \tau
}{%
    \tau \models_\Sigma \captures \phi
}
\quad
\inferrule{%
    \sigma \models_\Sigma \phi \\
    \sigma \composable^\omega \tau
}{%
    \tau \models_\Sigma \composable \phi
}
\]
in which $\sqsubseteq^\omega$ and $\composable^\omega$ are the pointwise extensions of the relations $\sqsubseteq$ and $\composable$, i.e., $\sigma \sqsubseteq^\omega \tau$ when, for all $n \in \naturals$, it holds that $\sigma(n) \sqsubseteq \tau(n)$, and similarly for $\composable^\omega$.

Although the atoms of our logic are formulas of the form $\phi = a \in \Sigma$ that have an exact matching semantics, in general one could use predicates over $\Sigma$. We chose not to do this to keep the presentation of examples simple.

The semantics of $\composable$ and $\captures$ match their descriptions: if $\sigma \in \Sigma^\omega$ is described by $\phi$ (i.e., $\sigma \models_\Sigma \phi$) and $\tau \in \Sigma^\omega$ captures this $\sigma$ at every action (i.e., $\sigma \sqsubseteq^\omega \tau$), then $\tau$ is a behavior described by $\captures \phi$ (i.e., $\tau \models_\Sigma \captures \phi$). Similarly, if $\rho \in \Sigma^\omega$ is described by $\phi$ (i.e., $\rho \models_\Sigma \phi$), and this $\rho$ is composable with $\sigma \in \sigma^\omega$ at every action (i.e., $\sigma \composable^\omega \rho$), then $\rho$ is described by $\composable \phi$ (i.e., $\rho \models_\Sigma \composable \phi$).

As usual, we obtain disjunction ($\phi \vee \psi$), implication ($\phi \implies \psi$), ``always'' ($\always \phi$) and ``eventually'' ($\eventually \phi$) from these connectives. For example, $\eventually \phi$ is defined as $\true \until \phi$, meaning that, if $\sigma \models_\Sigma \eventually \phi$, there exists an $n \in \naturals$ such that $\sigma^{(n)} \models_\Sigma \phi$. The operator $\composable$ has an interesting dual that we shall consider momentarily.

We can extend $\models_\Sigma$ to a relation between SCAs (with underlying c-semiring $\abscsemiring$ and CAS $\Sigma$) and formulas in $\mathcal{L}_\Sigma$, by defining $A \models_\Sigma \phi$ to hold precisely when $\sigma \models_\Sigma \phi$ for all $\sigma \in L(A)$. In general, we can see that fewer properties hold as the threshold $t$ approaches the lowest preference in its semiring, as a consequence of the fact that decreasing the threshold can only introduce new (possibly undesired) behavior. Limiting the behavior of an SCA to some desired behavior described by a formula thus becomes harder as the threshold goes down, since the set of behaviors exhibited by that SCA is typically larger for lower thresholds. 

We view the tradeoff between available behavior and verified properties as essential and desirable in the design of robust autonomous systems, because it represents two options available to the designer. On the one hand, she can make a component more accommodating in composition (by lowering the threshold, allowing more behavior) at the cost of possibly losing safety properties. On the other hand, she can restrict behavior such that a desired property is guaranteed, at the cost of possibly making the component less flexible in composition.

\paragraph{Example: no wasted moves}
Suppose we want to verify that the agent never misses an opportunity to take a snapshot of a new location. This can be expressed by
\[
\phi_\w = \captures \always (\move \implies \nxt (\neg \move \until \snapshot))
\]
This formula reads as ``every behavior captures that, at any point, if the current action is a move, then it is followed by a sequence where we do not move until we take a snapshot''. Indeed, if $t_e \otimes t_s = 5$, then $A_{\e, \s} \models_\Sigma \phi_\w$, since in this case every behavior of $A_{\e, \s}$ captures that between $\move$-actions we find a $\snapshot$-action. However, if $t_e \otimes t_s = 7$, then $A_{\e, \s} \not\models_\Sigma \phi_\w$, since $\angl{\move_2, \move_2, \charge, \charge, \charge, \charge}^\omega$ would be a behavior of $A_{\e, \s}$ that does not satisfy $\phi_\w$, as it contains two successive actions that capture $\move$.\footnote{Recall that $\move_2$ is the composition of $\move$ and $\discharge_2$, i.e., $\move \sqsubseteq \move_2$.} This shows the primary use of $\captures$, which is to verify the behavior of a component in terms of the behavior contributed by subcomponents.


\paragraph{Example: verifying a component interface}
Another application of the operator $\composable$ is to verify properties of the behavior composable with a component. Suppose we want to know whether all behaviors composable with a behavior of $A$ validate $\phi$. Such a property is useful, because it tells us that, in composition, $A$ filters out the behaviors of the other operand that do not satisfy $\phi$. Thus, if every behavior that composes with a behavior of $A$ indeed satisfies $\phi$, we know something about the behavior \emph{imposed} by $A$ in composition. Perhaps surprisingly, this use can be expressed using the $\composable$-connective, by checking whether $A \models_\Sigma \neg \composable \neg \phi$ holds; for if this is the case, then for all $\sigma, \tau \in \Sigma^\omega$ with $\sigma$ a behavior of $A$ and $\sigma \composable^\omega \tau$, we know that $\sigma \not\models_\Sigma \composable \neg \phi$, thus in particular $\tau \not\models_\Sigma \neg \phi$ and therefore $\tau \models_{\Sigma} \phi$.

More concretely, consider the component $A_\e$. From its structure, we can tell that the action $\charge$ must be executed at least once every five moves. Thus, if $\tau$ is composable with a behavior of $A_\e$, then $\tau$ must also execute some action composable with $\charge$ once every five moves. This claim can be encoded by
\[
\phi_\cc = \neg \composable \neg \always \left(\nxt \composable \charge \vee \nxt^2 \composable \charge \vee \dots \vee \nxt^5 \composable \charge \right)
\]
where $\nxt^n$ denotes repeated application of $\nxt$. If $A_\e \models_\Sigma \phi_\cc$, then every behavior of $A_\e$ is incomposable with behavior where, at some point, one of the next five actions is not composable with with $\charge$. Accordingly, if $\sigma \in \Sigma^\omega$ is  composable with some behavior of $A_\e$, then, at every point in $\sigma$, one of the next five actions must be composable with $\charge$. All behaviors that fail to meet this requirement are excluded from the composition.

\subsection{Decision procedure}%
\label{section:decision-procedure}

We developed a procedure to decide whether $A \models_\Sigma \phi$ holds for a given SCA $A$ and $\phi \in \mathcal{L}_\Sigma$. 
\iftechreport%
The full details of this procedure are given in Appendix~\ref{appendix:decision-procedure};
\fi%
the main results are summarized below.

\begin{proposition}%
\label{proposition:decision-procedure}
Let $\phi \in \mathcal{L}_\Sigma$. Given an SCA $A$ and CAS $\Sigma$, the question whether $A \models_\Sigma \phi$ is decidable. In case of a negative answer, we obtain a stream $\sigma \in \Sigma^\pi$ such that $\sigma \in L(A)$ but $\sigma \not\models_\Sigma \phi$. The total worst-case complexity is bounded by a stack of exponentials in $|\phi|$, i.e., 
\(
2^{.^{.^{.^{|\phi|}}}}
\),
whose height is the maximal nesting depth of $\captures$ and $\composable$ in $\phi$, plus one. 
\end{proposition}
This complexity is impractical in general, but we suspect that the nesting depth of $\captures$ and $\composable$ is at most two for almost all use cases. We exploit the counterexample in Section~\ref{section:diagnostics}.


\section{Diagnostics}%
\label{section:diagnostics}

Having developed a logic for SCAs as well as its decision procedure, we investigate how a designer can cope with undesirable behavior exhibited by the agent, either as a run-time behavior $\sigma$, or as a counterexample $\sigma$ to a formula found at design-time (obtained through Proposition~\ref{proposition:decision-procedure}). The tools outlined here can be used by the designer to determine the right threshold value for a component given the properties that the component (or the system at large) should satisfy.

\subsection{Eliminating undesired behavior}
A simple way to counteract undesired behavior is to see if the threshold can be raised to eliminate it --- possibly at the cost of eliminating other behavior. For instance, in Section~\ref{section:syntax-and-semantics}, we saw a formula $\phi_\w$ such that $A_{\e, \s} \not\models_\Sigma \phi_\w$, with counterexample $\sigma = \angl{\move_2, \move_2, \charge, \charge, \charge, \charge}^\omega$, when $t_\e \otimes t_\s = 7$. Since all $\move_2$-labeled transitions of $A_{\e, \s}$ have preference $7$, raising\footnote{Recall that $7 \leq_{\wcsemiring} 5$, so $5$ is a ``higher'' threshold in this context.} $t_\e \otimes t_\s$ to $5$ ensures that $\sigma$ is not present in $L(A_{\e, \s})$; indeed, if $t_\e \otimes t_\s = 5$, then $A_{\e, \s} \models_\Sigma \phi_\w$. We should be careful not to raise the threshold too much: if $t_\e \otimes t_\s = 0$, then $L(A_{\e, \s}) = \emptyset$, since every behavior of $A_{\e, \s}$ includes a transition with a non-zero weight --- with threshold $t_e \otimes t_\s = 0$, $A_{\e, \s} \models_\Sigma \psi$ holds for \emph{any} $\psi$.

In general, since raising the threshold does not add new behavior, this does not risk adding additional undesired behavior. The only downside to raising the threshold is that it possibly eliminates desirable behavior. We define the \emph{diagnostic preference} of a behavior as a tool for finding such a threshold.

\begin{definition}
Let $A = \angl{Q, \Sigma, \abscsemiring, \rightarrow, q^0, t}$ be an SCA, and let $\sigma \in \Sigma^\pi \cup \Sigma^*$. The \emph{diagnostic preference} of $\sigma$ in $A$, denoted $d_A(\sigma)$, is calculated as follows:
\begin{enumerate}
    \item Let $Q_0$ be $\{ q^0 \}$, and for $n < |\sigma|$ set $Q_{n+1} = \{ q' : q \in Q_n,\ q \myrightarrow{\sigma(n),\ e} q' \}$.
    \item Let $\xi \in \abscsemiring^\pi \cup \abscsemiring^*$ be the stream such that $\xi(n) = \bigoplus \{ e : q \in Q_n,\ q \myrightarrow{\sigma(n),\ e} q' \}$.
    \item $d_A(\sigma) = \bigwedge \{ \xi(n) : n \leq |\sigma| \}$, with $\bigwedge$ the greatest lower bound operator of $\abscsemiring$.
\end{enumerate}
\end{definition}

Since $\sigma$ is finite or eventually periodic, and $Q$ is finite, $\xi$ is also finite or eventually periodic. Consequently, $d_A(\sigma)$ is computable.

\begin{lemma}%
\label{lemma:diagnostic-preference-bound}
Let $A = \angl{Q, \Sigma, \abscsemiring, \rightarrow, q^0, t}$ be an SCA, and let $\sigma \in \Sigma^\pi \cup \Sigma^*$. If $\sigma \in L(A)$, or $\sigma$ is a finite prefix of some $\tau \in L(A)$, then $t \leq_\abscsemiring d_A(\sigma)$.
\end{lemma}
\begin{proof}
If $\sigma \in L(A)$, there exist streams $\mu \in Q^\omega$ and $\nu \in \abscsemiring^\omega$ such that $\mu(n) = q^0$, and for all $n \in \naturals$, $t \leq \nu(n)$ and $\mu(n) \myrightarrow{\sigma(n),\ \nu(n)} \mu(n+1)$. It is not hard to see that $\mu(n) \in Q_n$ for $n \in \naturals$. Then also $t \mathrel{\leq_\abscsemiring} \nu(n) \mathrel{\leq_\abscsemiring} \xi(n)$ for all $n \in \naturals$. Thus, $t \leq_\abscsemiring d_A(\sigma)$. Likewise, if $\sigma$ is a finite prefix of some $\tau \in L(A)$, then $t \leq_\abscsemiring d_A(\tau)$ by the above, and $d_A(\tau) \leq_\abscsemiring d_A(\sigma)$ by definition of $d_A$, thus $t \leq_\abscsemiring d_A(\sigma)$.
\end{proof}

Since $d_A(\sigma)$ is a necessary upper bound on $t$ when $\sigma$ is a behavior of $A$, it follows that we can exclude $\sigma$ from $L(A)$ if we choose $t$ such that $t \not\leq_\abscsemiring d_A(\sigma)$. In particular, if we choose $t$ such that $d_A(\sigma) <_\abscsemiring t$, then $\sigma \not\in L(A)$. Note that this may not always be possible: if $d_A(\sigma)$ is $\cstop$ then such a $t$ does not exist.

Note that there may be another threshold (i.e., not obtained by Lemma~\ref{lemma:diagnostic-preference-bound}), which may also eliminate fewer desirable behaviors. Thus, while this lemma gives helps to choose a threshold to exclude some behaviors, it is not a definitive guide. 
\iftechreport%
We refer to Appendix~\ref{appendix:caveat} for a concrete example.
\fi%

\subsection{Localizing undesired behavior}

One can also use the diagnostic preference to identify the components that are involved in allowing undesired behavior. Let us revisit the first example from Section~\ref{section:syntax-and-semantics}, where we verified that every pair of $\move$-actions was separated by at least one $\snapshot$ action, as described in $\phi_\w$. Suppose we choose $t_\e = 10$ and $t_\s = 1$; then $t_\e \otimes t_\s = 11$, thus $\sigma = \angl{\move_2, \charge, \charge}^\omega \in L(A_\s)$, meaning $A_{\e,\s} \not\models_\Sigma \phi_\w$. By Lemma~\ref{lemma:diagnostic-preference-bound}, we find that $11 = t_{\e, \s} = t_\e \otimes t_\s \leq_\wcsemiring d_{A_{\e,\s}}(\sigma) = 7$.
Even if $A_\s$'s threshold were as strict as possible (i.e., $t_\s = 0 = \cstop_\wcsemiring$), we would find that $t_\e \otimes t_\s \leq_\wcsemiring d_{A_{\e,\s}}(\sigma)$, meaning that we cannot eliminate $\sigma$ by changing $t_\s$ only. In some sense, we could say that $t_\e$ is responsible for $\sigma$.\footnote{Arguably, $A_\e$ as a whole may not be responsible, because modifying the preference of the $\move$-loop on $q_N$ in $A_\s$ can help to exclude the undesired behavior as well. In our framework, however, the threshold is a generic property of any SCA, and so we use it as a handle for talking about localizing undesired behaviors to component SCAs.} 

More generally, let ${(A_i)}_{i \in I}$ be a finite family of automata over the c-semiring $\abscsemiring$ with thresholds ${(t_i)}_{i \in I}$. Furthermore, let $A = \bigbowtie_{i \in I} A_i$ and let $\psi$ be such that $A \not\models_\Sigma \psi$, with counterexample behavior $\sigma$. Suppose now that for some $J \subseteq I$, we have \( \bigotimes\nolimits_{i \in J} t_i \leq_\abscsemiring d_A(\sigma) \). Since $\otimes$ is intensive, we furthermore know that \( \bigotimes\nolimits_{i \in I} t_i \leq_\abscsemiring \bigotimes\nolimits_{i \in J} t_i \). Therefore, at least one of $t_i$ for $i \in J$ must be adjusted to exclude the behavior $\sigma$ from the language of $\bigbowtie_{i \in I} A_i$.

We call ${(t_i)}_{i \in J}$ \emph{suspect} thresholds: \emph{some} $t_i$ for $i \in I$ must be adjusted to eliminate $\sigma$; by extension, we refer to $J$ as a \emph{suspect subset} of $I$. Note that $I$ may have distinct and disjoint suspect subsets.
If $J \subseteq I$ is disjoint from every suspect subset of $I$, then $J$ is called \emph{innocent}. If $J$ is innocent, changing $t_j$  for some $j \in J$ (or even $t_j$ for all $j \in J$) alone does not exclude $\sigma$. Finding suspect and innocent subsets of $I$ thus helps in finding out which thresholds need to change in order to exclude a specific undesired behavior.
\begin{algorithm}
\SetKwFunction{FindSuspect}{FindSuspect}
\SetKwProg{Fn}{Function}{:}{end}
\Fn{\FindSuspect(I)}{%
    $M := \emptyset$\;
    \ForEach{$i \in I$}{%
       \If{$I \setminus \{ i \}$ is suspect}{%
           $M := M \cup \FindSuspect(I \setminus \{ i \})$\;
       }
    }
    \eIf{$M = \emptyset$}{%
       \Return{$\{ I \}$}\;
    }{%
       \Return{$M$}\;
    }
}
\caption{Algorithm to find minimal suspect subsets.}\label{figure:algorithm-find-minimal-suspect-subsets}
\end{algorithm}

Algorithm~\ref{figure:algorithm-find-minimal-suspect-subsets} gives pseudocode to find minimal suspect subsets of a suspect set $I$; we argue correctness of this algorithm in Theorem~\ref{theorem:minimal-suspect-subsets-correctness}; for a proof, see~\cite{kappe-arbab-talcott-2017-techreport}.

\begin{theorem}%
\label{theorem:minimal-suspect-subsets-correctness}
If $I$ is suspect and $d_A(\sigma) < \cstop$, then $\mathtt{FindSuspect}(I)$ contains exactly the minimal suspect subsets of $I$. 
\end{theorem}
\begin{proof}
First, note that it is easy to see that $\mathtt{FindSuspect}$ never returns $\emptyset$.

The proof proceeds by induction on $I$. In the base, where $I = \{ i \}$, we can see that $\bigotimes \emptyset = \cstop$, thus, since $d_A(\sigma) < \cstop$, it follows that $I \setminus \{ i \} = \emptyset$ is not suspect. The first branch of the subsequent $\mathbf{if}$ is selected, which returns $\{ I \}$ itself. This matches the fact that $I$ is the only suspect subset of $I$.

In the inductive step, we assume the claim holds for all strict subsets of $I$. We consider two cases. On the one hand, if there exists an $i \in I$ such that $I \setminus \{ i \}$ is suspect, then we know that the $\mathbf{foreach}$-loop will modify $M$ (since $\mathtt{FindSuspect}$ never returns an empty set). Moreover, $I$ itself is not minimally suspect. The algorithm then returns 
\[\bigcup \{ \mathtt{FindSuspect}(I \setminus \{ i \}) : i \in I,\ I \setminus \{ i \}\ \mathrm{suspect} \}\]
By induction, $\mathtt{FindSuspect}(I \setminus \{ i \})$ returns all minimal suspect subsets of $I \setminus \{ i \}$. Since each of these is also a minimal suspect subset of $I$, and since very minimal suspect subset of $I$ that is not equal to $I$ is contained in one of these, the claim follows by the fact that we ruled out $I$ as a minimal suspect subset.
\end{proof}
In the case where $d_A(\sigma) = \cstop$, it is easy to see that $\{ \{ i \} : i \in I \}$ is the set of minimal suspect subsets of $I$.

In the worst case, every subset of $I$ is suspect, and therefore the only minimal suspect subsets are the singletons; in this scenario, there are $O(|I|!)$ calculations of a composed threshold value. Using memoization to store the minimal suspect subsets of every $J \subseteq I$, the complexity can be reduced to $O(2^{|I|})$. 

While this complexity makes the algorithm seem impractical ($I$ need not be a small set), we note that the case where all components are individually responsible for allowing a certain undesired behavior should be exceedingly rare in a system that was designed with the violated concern in mind: it would mean that \emph{every component} contains behavior that ultimately composes into the undesired behavior --- in a sense, \emph{facilitating} behavior that counteracts their interest. 

\section{Discussion}%
\label{section:conclusion}

In this paper, we proposed a framework that facilitates the construction of autonomous agents in a compositional fashion. We furthermore considered an LTL-like logic for verification of the constructed models that takes their compositional nature into account, and showed the added value of operators related to composition in verifying properties of the interface between components. We also provided a decision procedure for the proposed logic. 

The proposed agents are ``soft'', in that their actions are given preferences, which may or may not make the action feasible depending on the threshold preference. The designer can \emph{decrease} this threshold  to allow for more behavior, possibly to accommodate the preferences of another component, or \emph{increase} it to restrict undesired behavior observed at run-time or counterexamples to safety assertions found at design-time. We considered a simple method to raise the threshold enough to exclude a given behavior, but which may overapproximate in the presence of partially ordered preferences, possibly excluding desired behavior. 

In case of a composed system, one can also find out which component's thresholds can be thought of as \emph{suspect} for allowing a certain behavior. This information can give the designer a hint on how to adjust the system --- for example, if the threshold of an energy management component turns out to be suspect for the inclusion of undesired behavior, perhaps the component's threshold needs to be more conservative with regard to energy expenses to avoid the undesired behavior. We stress that responsibility may be assigned to a \emph{set} of components as a whole, if their composed threshold is suspect for allowing the undesired behavior, which is possible when preferences are partially ordered.

\section{Further Work}%
\label{section:further-work}

Throughout our investigation, the tools for verification and diagnosis were driven by the compositional nature of the framework. As a result, they apply not only to the ``grand composition'' of all components of the system, but also to subcomponents (which may themselves be composed of sub-subcomponents). What is missing from this picture is a way to ``lift'' verified properties of subcomponents to the composed system, possibly with a side condition on the interface between the subcomponent where the property holds and the subcomponent representing the rest of the system, along the lines of the interface verification in Section~\ref{section:syntax-and-semantics}.

If we assume that agents have low-latency and noiseless communication channels, one can also think of a multi-agent system as the composition of SCAs that represent each agent. As such, our methods may also apply to verification and diagnosis of multi-agent systems. However, this assumption may not hold. One way to model this could be to insert ``glue components'' that mediate the communication between agents, by introducing delay or noise. Another method would be to introduce a new form of composition for loosely coupled systems.


Finding an appropriate threshold value also deserves further attention. In particular, a method to adjust the threshold value \emph{at run-time}, would be useful, so as to allow an agent to relax its goals as gracefully as possible if its current goal appears unachievable, and raise the bar when circumstances improve.


Lastly, the use soft constraints for autonomous agents is also being researched in a parallel line of work~\cite{talcott-arbab-yadav-2015}, which employs rewriting logic. Since rewriting logic is backed by powerful tools like Maude, with support for soft constraints~\cite{wirsing-etal-2007}, we aim to reconcile the automata-based perspective with rewriting logic.

\iftechreport%

\begin{appendix}

\section{Decision Procedure}%
\label{appendix:decision-procedure}

In this appendix, we work out the details of a decision procedure for the logic proposed in Section~\ref{section:linear-temporal-logic}, i.e., a procedure to decide whether $A \models_\Sigma \phi$ holds for a given $A$ and $\phi$. This method follows~\cite{vardi-1995}, i.e., we do the following:
\begin{enumerate}
    \item Translate $A$ to a \buchi-automaton $A_M$ with the same language as $A$. 
    \item Translate $\phi$ to a \buchi-automaton $A_\phi$ that accepts the streams verified by $\phi$.
    \item Check whether language of $A_M$ is contained in that of $A_\phi$.
\end{enumerate}

The last step is an instance of checking $\omega$-regular language containment, which can be decided in $O(2^{|A_\phi|})$, where $|A_\phi|$ is the number of states of $A_\phi$~\cite{vardi-1995}. Moreover, in case of a negative answer, this method provides a $\sigma \in \Sigma^\pi$ such that $\sigma \in L(A_M)$ but $\sigma \not\in L(A_\phi)$, and therefore $\sigma \in L(A)$ but $\sigma \not\models_\Sigma \phi$.

We give the details for the first two steps step below, but first we briefly recall the details of \buchi-automata.

\subsection{\buchi-automata}

A \emph{(non-deterministic) \buchi-automaton}~\cite{buchi-1962} (BA) is a tuple $A = \angl{Q, \Sigma, \rightarrow, q^0, F}$ such that $Q$ is a finite set of \emph{states}, with $q^0 \in Q$ the \emph{initial state} and $F \subseteq Q$ the set of \emph{accepting states}, $\Sigma$ is a finite set called the \emph{alphabet} and $\rightarrow\ \subseteq Q \times \Sigma \times Q$ is a relation called the \emph{transition relation}. We write $q \myrightarrow{a} q'$ whenever $\angl{q, a, q'} \in\ \rightarrow$. 

A stream $\lambda \in Q^\omega$ is a \emph{trace} of a stream $\sigma \in \Sigma^\omega$ in $A$ if $\lambda(n) \myrightarrow{\sigma(n)} \lambda(n+1)$ holds for all $n \in \naturals$. A trace $\lambda$ is \emph{accepting} if $\lambda(n) \in F$ for infinitely many $n \in \naturals$. A stream $\sigma \in \Sigma^\omega$ is \emph{accepted} by $A$ if it has an accepting trace $\lambda$ such that $\lambda(0) = q^0$. The set of streams accepted by $A$ is the \emph{language} of $A$ and denoted by $L(A)$.

An \emph{Alternating \buchi-automaton} (ABA) is a tuple $A = \angl{Q, \Sigma, \rightarrow, q^0, F}$ such that $Q$ is a finite set of \emph{states} with $q^0 \in Q$ the \emph{initial state} and $F \subseteq Q$ the set of \emph{accepting states}, $\Sigma$ is a finite set called the \emph{alphabet} and $\rightarrow \ \subseteq Q \times \Sigma \times 2^Q$ is a (finite) relation called the transition relation. Unlike a BA, a single transition in an ABA can have multiple destinations. We write $q \myrightarrow{a} P$ when $\angl{q, a, P} \in\ \rightarrow$. A \emph{run} of a stream $\sigma \in \Sigma^\omega$ in $A$ is a labeled tree $T$ such that the root of $T$ is labeled with $q^0$, and when $q$ is the label of a node of $T$ at depth $n$ and the set of labels of children of said node is $P$, $q \myrightarrow{\sigma(n)} P$ is a transition of $A$. A run $T$ is \emph{accepting} if every infinite branch of $T$ is labeled by an accepting state infinitely often.

ABAs accept the same languages as their non-deterministic cousins~\cite{miyano-hayashi-1984}: given an ABA $A$, we can construct a BA $A'$ such that $L(A) = L(A')$.

\subsection{SCAs to a BAs}

The translation of an SCA to a BA is relatively straightforward.

\begin{lemma}%
\label{lemma:sca-to-buchi}
Let $A$ be an SCA\@. We can construct a BA $A'$ such that $L(A) = L(A')$.
\end{lemma}
\begin{proof}
Choose $A' = \angl{Q, \Sigma, \rightarrow_t, q^0, Q}$, where $\rightarrow_t$ is the relation in which $q \myrightarrow{a}_t q'$ if and only if $q \myrightarrow{a,\, e} q'$ and $t \leq e$. We can now use the witness streams for $\sigma \in L(A)$ to show that $\sigma \in L(A')$ and vice versa. Indeed, for the inclusion from right to left we can infer the existence of a stream $\nu \in \abscsemiring^\omega$, while for the inclusion from left to right we simply discard the stream of preferences.
\end{proof}

\subsection{Formulas to BAs}

We present two methods to translate a formula into a BA that accepts precisely the streams validated by the formula. The first method is an extension of the recursive translation by Sherman et al.~\cite{sherman-pnueli-harel-1984}. We also propose an extension to the approach of Muller et al.~\cite{muller-saoudi-schupp-1988}, which has a different complexity bound.

\subsubsection{Recursive method}
One easily constructs BAs that represent atomic formulas $\top$ and $\sigma$. Moreover, we can recreate the effect of logical connectives using BAs; for example, we can construct $A_{\captures}$ such that $\tau \in L(A_{\captures})$ if and only if there exists a $\sigma \in L(A)$ with $\sigma \sqsubseteq^\omega \tau$ easily: simply choose $A_{\captures} = \angl{Q, \Sigma, \rightarrow_{\captures}, q^0, F}$ where $q \myrightarrow{b,\ e}_{\captures} q'$ if and only if $q \myrightarrow{a,\ e} q'$ with $a \sqsubseteq b$. Similar constructions exist for the other connectives, including $\composable$. This is formalized in the following lemma.

\begin{lemma}%
\label{lemma:connectives-on-automata}
Let $A_1$ and $A_2$ be BAs over alphabet $\Sigma$ and let $a \in \Sigma$. One can construct BAs $A_a$, $A_\wedge$, $A_{\until}$, $A_{\nxt}$, $A_\neg$, $A_{\captures}$ and $A_{\composable}$ such that the following are true
\begin{enumerate}[(i)]
    \item\label{claim:ltl-to-buchi-atomic} $\sigma \in L(A_a)$ if and only if $\sigma(0) = a$
    \item\label{claim:ltl-to-buchi-conjunction} $\sigma \in L(A_\wedge)$ if and only if $\sigma \in L(A_1)$ and $\sigma \in L(A_2)$~\cite{choueka-1974}
    \item\label{claim:ltl-to-buchi-until} $\sigma \in L(A_{\until})$ if and only if there exists an $n \in \naturals$ such that for all $k < n$, $\sigma^{(k)} \in L(A_1)$ and $\sigma^{(n)} \in L(A_2)$
    \item\label{claim:ltl-to-buchi-next} $\sigma \in L(A_{\nxt})$ if and only if $\sigma' \in L(A_1)$
    \item\label{claim:ltl-to-buchi-negation} $\sigma \in L(A_\neg)$ if and only if $\sigma \not\in L(A_1)$~\cite{buchi-1962}
    \item\label{claim:ltl-to-buchi-captures} $\sigma \in L(A_{\captures})$ if and only if $\tau \sqsubseteq^\omega \sigma$ for some $\tau \in L(A_1)$
    \item\label{claim:ltl-to-buchi-composable} $\sigma \in L(A_{\composable})$ if and only if $\sigma \composable^\omega \tau$ for some $\tau \in L(A_1)$
\end{enumerate}
\end{lemma}
\begin{proof}
We treat the claims one by one.
\begin{enumerate}[(i)]
    \item The BA $A = \angl{\{ q^0, q^1 \}, \Sigma, \rightarrow, q^0, \{ q^1 \}}$, where $\rightarrow$ is smallest relation such that $q^0 \myrightarrow{a} q^1$ and $q^1 \myrightarrow{b} q^1$ for $b \in \Sigma$, suffices.
    \item Refer to~\cite[Proposition 6]{vardi-1995} for a proof.
    \item Let $A_i = \angl{Q_i, \Sigma, \rightarrow_i, q^0_i, F_i}$ for $i \in \{ 1, 2 \}$ and assume (without loss of generality) that $Q_1$ and $Q_2$ are disjoint. We choose $Q = Q_1 \cup Q_2 \cup \{ q^0 \}$ and $F = F_1 \cup F_2$, and let $\rightarrow$ be the smallest relation in $Q \times \Sigma \times 2^Q$ satisfying the rules
    \[
    \inferrule{%
        i \in \{ 1, 2 \} \\
        q \myrightarrow{a}_i q'
    }{%
        q \myrightarrow{a} \{ q' \}
    }
    \quad
    \inferrule{%
        q^0_1 \myrightarrow{a}_1 q'
    }{%
        q^0 \myrightarrow{a} \{ q', q^0 \}
    }
    \quad
    \inferrule{%
        q^0_2 \myrightarrow{a}_2 q''
    }{%
        q^0 \myrightarrow{a} \{ q'' \}
    }
    \]
    We choose $A_{\until} = \angl{Q, \Sigma, \rightarrow, q^0, F}$. It remains to show that $A_{\until}$ validates the claim. If $\sigma \in L(A_{\until})$, then there is an accepting run $T$ of $\sigma$ in $A_{\until}$. Let us refer to the nodes of $T$ labeled with $q^0$ as \emph{pivot nodes}. For every $n \in \naturals$, there is at most one pivot node at depth $n$, since every pivot node has a pivot node as its parent, and at most one pivot node among its children. Furthermore, there are two types of pivot nodes:
    \begin{itemize}
        \item nodes with children labeled by $q^0$ and some $q' \in Q_1$, called \emph{branch nodes}
        \item nodes with children labeled by some $q'' \in Q_2$, called \emph{stop nodes}
    \end{itemize}
    All children of stop nodes must be labeled with states in $Q_2$; as a consequence, no stop node has a pivot node in its descendants. This shows that there is at most one stop node in $T$. Furthermore, if there were no stop nodes in $T$, then every pivot node would have a branch node among its children, rendering an infinite run of pivot nodes, which contradicts that $T$ is an accepting run. We can thus derive that $T$ has exactly one stop node at some depth $n$, and that all other pivot nodes in $T$ occur as branch node parents of this stop node. One can then show that if $k < n$, we can construct an accepting run for $\sigma^{(k)}$ in $A_1$ from the subtree of $T$ rooted at the unique branch node with depth $k$, and that the tree rooted at the stop node gives us an accepting run of $\sigma^{(n)}$ in $A_2$.

    For the other direction, one can combine the accepting runs of $\sigma^{(k)}$ in $A_1$ and $\sigma^{(n)}$ in $A_2$ into an accepting run of $\sigma$ in $A_{\until}$ easily, by starting with a finite tree consisting of $n$ pivot nodes, and attaching the runs.

    \item Simply add a new state $q$ to $A_1$ and make this the initial state, then add a transition from $q$ to the initial state of $A_1$ labeled with $a$ for every $a \in \Sigma$. It follows that $\sigma \in L(A_X)$ if and only if $\sigma' \in L(A_1)$.
    \item Refer to~\cite{sistla-vardi-wolper-1985} or~\cite{safra-1988} for a proof.
    \item Let $A_1 = \angl{Q_1, \Sigma, \rightarrow_1, q^0_1, F_1}$. Choose $A_{\captures} = \angl{Q_1, \Sigma, \rightarrow_{\captures}, q^0_1, F_1}$, where $q \myrightarrow{b}_{\captures} q'$ if and only if there exists an $a \in \Sigma$ such that $q \myrightarrow{a}_1 q'$ and $a \sqsubseteq b$. It is easily shown that $A_{\captures}$ validates the claim.
    \item By a construction analogous to the above. \qedhere%
\end{enumerate}
\end{proof}

We thus obtain a recursive formula-to-automaton translation; for example, if $\phi = \neg \psi$, we construct the automaton $A_\psi$ representing $\psi$, from which we obtain the automaton $A_\phi$ representing $\phi$ by the aforementioned construction.
\begin{corollary}%
\label{corollary:recursive-construction}
Let $\phi \in \mathcal{L}_\Sigma$. Then one can construct an automaton $A_\phi$ such that $\sigma \models_\Sigma \phi$ if and only if $\sigma \in L(A_\phi)$.
\end{corollary}

\paragraph{Complexity}
The construction in Corollary~\ref{corollary:recursive-construction} sees a sharp rise in the number of states at each recursion. For instance, in the construction of $A_\wedge$ referenced above, if $A_1$ and $A_2$ have $n$ and $m$ states respectively, then $A_\wedge$ has $2nm$ states~\cite{choueka-1974}. The situation is worse for negation; here, we see a necessarily exponential rise in the number of states~\cite{safra-1988}. Thus the recursive translation ends up with the prohibitively unfeasible upper bound of a stack of exponentials, which is as high as the nesting depth of negations in the formula.

\subsubsection{Subformula construction}

Another approach~\cite{muller-saoudi-schupp-1988} translates $\phi$ to a language-equivalent ABA, where each (possibly negated) subformula is a state. Intuitively, a state representing subformula $\psi$ can be seen as a requirement that the remainder of the satisfies $\psi$. This approach yields an automaton of size linear in the size of the formula, but the translation from ABA to BA is necessarily exponential~\cite{boker-kupferman-rosenberg-2010}.

To use this method for our logic, we need to incorporate the connectives $\captures$ and $\composable$, i.e., we need to interlink a state representing a subformula of the form $\captures \psi$ to other states, such that the streams accepted starting in that state are the streams that validate $\captures \psi$. The na\"{\i}ve way to do this is to extend the state space of our output automaton to include subformulas of the input formula $\phi$ under the $\captures$ connective. For instance, a formula of the form $\phi = \captures (\psi_1 \wedge \psi_2)$ would give states for $\phi$, $\psi_1 \wedge \psi_2$, $\psi_1$, $\psi_2$, $\captures \psi_1$ and $\captures \psi_2$. If $a \sqsubseteq b$ and the state $\psi$ has a transition of the form $\psi \myrightarrow{a} P$, then we add a transition $\captures \psi \myrightarrow{b} \{ \captures \rho : \rho \in P \}$.

Unfortunately, this approach is unsound in general. For example, consider a CAS $\Sigma$ with distinct actions $a$, $b$ and $c$ such that $a, b \sqsubseteq c$, and set $\phi = \captures \nxt (a \wedge b)$. Note that there exists no $\sigma \in \Sigma^\omega$ with $\sigma \models_\Sigma \phi$. We now try to translate $\phi$ to a semantically equivalent ABA using the construction above. By the subformula construction, we find a BA $A_{\nxt (a \wedge b)}$ representing $\nxt (a \wedge b)$, with transitions $\phi \myrightarrow{x} \{ a,  b \}$ for $x \in \Sigma$, as well as $a \myrightarrow{a} \emptyset$ and $b \myrightarrow{b} \emptyset$, as a subautomaton of $A_\phi$. For the remainder of $A_\phi$, we have states $\phi$, $\captures a$ and $\captures b$ with transitions $\phi \myrightarrow{x} \{ \captures a, \captures b \}$ for $x \in \Sigma$, and $\captures a \myrightarrow{a} \emptyset$, $\captures a \myrightarrow{c} \emptyset$ as well as $\captures b \myrightarrow{b} \emptyset$ and $\captures b \myrightarrow{c} \emptyset$ (since $a \sqsubseteq a, c$ and $b \sqsubseteq b, c$). We can now construct a tree with a root labeled by $\phi$, and two children labeled by $\captures a$ and $\captures b$ respectively, as a run showing that $\angl{c}^\omega \in L(A_\phi)$. A similar pathological case exists for the operator $\composable$.

The essence of the problem above is in the use of ABA, where lifting the construction for $\captures$ in BAs is unsound. Specifically, we have an accepting run for the behavior $\angl{c}^\omega$ starting at $\captures \nxt (a \wedge b)$ in the form of a tree $T$, but this run does not give rise to an accepting run $T'$ starting at $\nxt (a \wedge b)$. In general, if $\captures \phi \myrightarrow{z} \captures \phi'$ and $\captures \phi \myrightarrow{z} \captures \phi''$ then the construction only guarantees that there exist $x, y$ such that $\phi \myrightarrow{x} \phi'$ and $\phi \myrightarrow{y} \phi''$ with $z \sqsubseteq x, y$, while $x$ and $y$ may differ.

We briefly sketch a method that gets around this problem. Instead of the above, we can use the subformula construction as follows. Given $\phi$, find all subformulas of the form $\captures \psi$ or $\composable \chi$ which do not appear below $\captures$ or $\composable$. Recursively construct the ABAs representing $A_\psi$ and $A_\chi$ for each of these, and convert them to equivalent BAs $A_\psi'$ and $A_\chi'$, before applying the (sound) conversion to BAs $A_{\captures \psi}$ and $A_{\composable \chi}$ representing $\captures \psi$ and $\composable \chi$ respectively. Now apply the subformula construction to $\phi$, except that the states representing $\captures \psi$ and $\composable \chi$ are replaced with the states of $A_{\captures \psi}$ and $A_{\composable \chi}$. The resulting automaton $A_\phi$ represents $\phi$; this can be shown by proving that if $q$ is a state of $A_\phi$ representing a subformula $\rho$ which is not below a $\captures$ or $\composable$, then the streams accepted at $q$ are precisely the streams that validate $\rho$; one can do this by induction on the structure of $\phi$, with atomic formulas $\top$ and $a$ for $a \in \Sigma$ as well as formulas of the form $\captures \psi$ and $\composable \chi$ as the base.

\paragraph{Complexity}
Due to the intermittent conversion of ABA to BA in the method outlined above, we can surmise that the complexity is bound from above by a stack of exponentials as high as the nesting depth of $\composable$ and $\captures$, plus one for the final translation of ABA to BA\@.

\section{Caveat regarding diagnostic preference}%
\label{appendix:caveat}

In this appendix, we show that applying the method that arises from Lemma~\ref{lemma:diagnostic-preference-bound} does not always give the lowest threshold that excludes a given behavior.

First, we fix the c-semiring $\abscsemiring$ as $\wcsemiring \times \wcsemiring$, that is: the carrier is ${(\mathbb{R} \cup \{ \infty \})}^2$, $\oplus_\abscsemiring$ is the pairwise minimum and $\otimes_\abscsemiring$ is the pairwise (affinely extended) sum, and furthermore $\angl{\infty, \infty}$ and $\angl{0, 0}$ are the minimal, respectively maximal elements. As a result, $\leq_\abscsemiring$ is the product order (i.e., $\angl{e_1, e_2} \leq_\abscsemiring \angl{e_1', e_2'}$ if and only if $e_1 \geq e_1'$ and $e_2 \geq e_2'$), and $\wedge_\abscsemiring$ is the pairwise maximum.

Furthermore, let $A$ be the SCA depicted below, with CAS $\Sigma = \{ a \}$.\footnote{The precise choice of $\composable$ and $\compose$ does not matter.}
\begin{figure}
    \centering
    \begin{tikzpicture}
        \begin{scope}[every node/.style={draw,circle},minimum size=10mm]
            \node (e0) at (0,0) {$q_0$};
            \node (e1) at (3,0) {$q_1$};
            \node (e2) at (-3,0) {$q_2$};
        \end{scope}

        \node[above=5mm of e0] (s) {};
        \path[->] (s) edge (e0);

        \path[->] (e0) edge node[above] {$a, \angl{4, 2}$} (e1);
        \path[->] (e0) edge node[above] {$a, \angl{2, 4}$} (e2);
        \path[->] (e1) edge[loop right] node[right] {$a, \angl{0, 0}$} (e1);
        \path[->] (e2) edge[loop left] node[left] {$a, \angl{0, 0}$} (e2);
    \end{tikzpicture}
\end{figure}

Suppose we want to choose $t$ such that $\sigma = \angl{a}^\omega$ is not in $L(A)$. We calculate:
\[
d_A(\sigma) = (\angl{2, 4} \oplus_\abscsemiring \angl{4, 2}) \wedge_\abscsemiring (\angl{0, 0} \oplus_\abscsemiring \angl{0, 0}) = \angl{2, 2} \wedge_\abscsemiring \angl{0, 0} = \angl{2, 2}
\]

By Lemma~\ref{lemma:diagnostic-preference-bound}, we know that if $\sigma \in L(A)$, then $t \leq_\abscsemiring \angl{2, 2}$. We can thus choose $t$ such that $t \not\leq_\abscsemiring \angl{2, 2}$ in order to exclude $\sigma$; for example, $t = \angl{1,1}$ would do. However, we can also choose $t = \angl{3, 3} \leq_\abscsemiring \angl{2, 2}$. In this case, we find that $\sigma \not\in L(A)$ as well. In conclusion, the application of Lemma~\ref{lemma:diagnostic-preference-bound} did not give the lowest threshold that excluded the given behavior.

\end{appendix}

\fi

\bibliographystyle{splncs03}
\bibliography{bibliography}

\end{document}